\documentclass[12pt, english, a4paper,reqno]{amsart}

\usepackage{graphicx}
\usepackage{tikz}
\usetikzlibrary{positioning,arrows}
\usetikzlibrary{arrows.meta}

\usepackage{capt-of}
\usepackage[applemac]{inputenc}
\usepackage{hyperref}
\usepackage{graphicx}

\usepackage{lscape}
\usepackage{amsmath,amssymb}
\usepackage{bm}
\usepackage{euscript,color}
\usepackage{babel}
\usepackage{amsthm}
\usepackage{amsfonts}
\usepackage{latexsym}
\usepackage{fancyhdr}
\usepackage{enumerate}

\usepackage{dsfont}
\usepackage{color}
\usepackage{mathtools}

\usepackage[numbers]{natbib}

\usepackage[left=2.5cm, top=3cm, bottom=3cm, right=2.5cm]{geometry}
\usepackage{soul}

\usepackage{setspace}
\usepackage{cancel}

\newcommand{\R}{\mathbb{R}}

\newcommand{\Chi}{\mathfrak{X}}

\newcommand{\comment}[1]{}

\makeatletter
\renewcommand{\section}{\@startsection%
{section}
{1}
{0mm}
{1.5\bigskipamount}
{0.5\bigskipamount}
{\centering\normalsize\sc}}

\renewcommand{\paragraph}{\@startsection%
{paragraph}
{4}
{0mm}
{\bigskipamount}
{-1.25ex}
{\normalsize\sl}}

\def\provedboxcontents#1{$\square$}

\makeatother

\newtheoremstyle{thm}{}{}{\slshape}{}{\scshape}{.}{0.5em}{}

\newtheoremstyle{def}{}{}{}{}{\scshape}{.}{0.5em}{}

\newtheoremstyle{rmk}{}{}{}{}{\scshape}{.}{0.5em}{}

\newtheoremstyle{claim}{}{}{}{}{\slshape}{.}{0.5em}{}

\newtheorem{newstatement}{newstatement}

\newtheorem{theorem}[newstatement]{Theorem}
\newtheorem{corollary}[newstatement]{Corollary}

\newtheorem*{conjecture*}{Conjecture}
\newtheorem*{prop*}{Proposition}
\newtheorem{definition}{Definition}

\theoremstyle{def}

\theoremstyle{rmk}

\theoremstyle{claim}

\expandafter\let\expandafter\oldproof\csname\string\proof\endcsname
\let\oldendproof\endproof
\renewenvironment{proof}[1][\proofname]{%
  \oldproof[\slshape #1]%
}{\oldendproof}

\let\leq\leqslant

\let\phi\varphi
\let\epsilon\varepsilon

\renewcommand{\emph}[1]{{\slshape #1}}
\renewcommand{\em}{\sl}

\title{Equilibrium selection: a geometric approach}

\author{Andrea Loi}
\address{Andrea Loi, Dipartimento di Matematica e Informatica \\
         Universit\`a di Cagliari, Italy.}
         \email{loi@unica.it}

\author{Stefano Matta}
\address{Stefano Matta, Dipartimento di Scienze economiche e Aziendali \\
         Universit\`a di Cagliari, Italy.}
         \email{smatta@unica.it}
         
\author{Daria Uccheddu}
\address{Daria Uccheddu, Dipartimento di Matematica e Informatica\\
         Universit\`a di Cagliari, Italy.}
         \email{daria.uccheddu@unica.it}

\date{\today}

\thanks{The first author  was supported  by Prin 2015 -- Real and Complex Manifolds; Geometry, Topology and Harmonic Analysis -- Italy, by INdAM. GNSAGA - Gruppo Nazionale per le Strutture Algebriche, Geometriche e le loro Applicazioni.}

\begin{document}
\begin{abstract}
In this paper we propose a geometric approach to the selection of the equilibrium price.
After a perturbation of the parameters, the new price is selected thorough the composition of two maps: the projection on the linearization of the equilibrium manifold, a method that underlies econometric modeling, and the exponential map, that associates a tangent vector with a geodesic on the manifold. 
As a corollary of our main result, we prove the equivalence between zero curvature and uniqueness of equilibrium in the case of an arbitrary number of goods and two consumers, thus extending the previous result by \cite{lmcurvuni}.


{\it{Subj.Class}}: {C61, C65, D50, D51.}

\end{abstract}
\maketitle

\vspace{0.3in}

\section{Introduction}\label{introduction}

In a pure exchange economy
with two consumers, an arbitrary number $l$ of goods and
fixed total resources $r$ distributed across agents,
suppose that at the initial endowment allocation
the set of corresponding equilibrium prices is not singleton, that is,
there are multiple equilibria.
As endowments vary, price adjusts towards a new equilibrium.
One could explore the out-of-equilibrium dynamics of this adjustment
or, less ambitiously, focus on a continuous approximation
represented by a sequence of equilibrium changes.
Yet this simpler approach, if there is price multiplicity, raises the crucial issue of price selection, that is, of which price will occur after a change in parameters.

In this setting the equilibrium price is often assumed to vary continuously (smoothly) as a result of a continuous (smooth) variation of parameters and
discontinuities in prices are usually attributed to singularities, a property called {\em smooth selection}, despite the lack of a theory behind it and the presence of other potential prices that could occur.

In the present paper we tackle this issue following the equilibrium manifold approach which, from the seminal work of Debreu, found an elegant geometric formulation in Balasko. We refer the reader to \cite{balib}.
We recall that the equilibrium manifold is defined as
the pairs of prices and endowments such that aggregate excess demand is zero.

The following figure is a stylized representation of the equilibrium manifold $E(r)$
and the smooth selection.
It depicts a pure exchange economy with multiple prices, where $\Omega(r)$
and $S$ denote the space of endowments and normalized prices, respectively.

\begin{figure}[h]
\begin{tikzpicture}
  
\draw[<-> ,thick] (7,0) node[below]{$\Omega(r)$} -- (0,0) --
    (0,6) node[left]{$S$};
    
\draw[rounded corners] (0.5,5) to [out=0,in=160] (5,4.5) 
 to [out=320,in=35] (1.5,2)
		    to [out=-100, in=170] (6.5,0.4);
\draw[draw=red, thick] (3,4.95) to [out=0,in=170] (4,4.8);		    

\node at (3,-0.26) {$\omega$};
\node at (4,-0.2) {$\omega'$};
  
\draw[draw=blue, thick] (3,0) to  (4,0);

\draw[fill] (3,0) circle [radius=0.05];
\draw[fill] (4,0) circle [radius=0.05];
\draw[fill] (3,4.95) circle [radius=0.05] node[above]{$x$};
\draw[fill] (4,4.8) circle [radius=0.05]  node[above]{$x'$};
\draw[fill] (4,3.1) circle [radius=0.05]  node[above right]{$y$};
\draw[fill] (0,3.1) circle [radius=0.05]  node[left]{$\tilde p$};

\draw[fill] (0,3.1) circle [radius=0.05]  node[left]{$\tilde p$};

\draw[dash dot] (5.35,4) -- (5.35,0);
\node at (5.35,-0.26) {$\omega_c$};

\draw [dash dot] (3,0)-- (3,4.95);
\draw[dash dot] (4,0) --(4,4.8);
\draw[dash dot] (0,3.1) --(4,3.1);

\end{tikzpicture}
\caption{A redistribution of endowments with smooth selection.}
\label{figsel}
\end{figure}
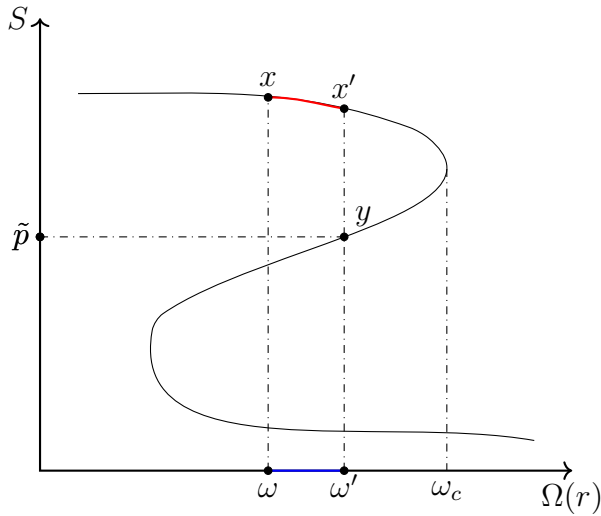

Starting from
$x=(p,\omega)$, we assume that the endowments are continuously
changed to $\omega'$. We observe that $\omega$ and $\omega'$ belong to a
region of the parameters with
three equilibria. The red curve represents the equilibrium path, i.e., prices and endowments consistent with the equilibrium during the transition of the economy from $\omega$ to $\omega'$.   Such a path is an approximation of the dynamics of price adjustment that ignores the out-of-equilibrium process. In fact, this curve is a subset of the equilibrium manifold.
The blue curve, an approximation of a discrete sequence of changes of the endowments, is the projection of the red curve onto the space of endowments and it is the cause of the price adjustment.

Following the smooth selection, it is generally accepted that this continuous (smooth) variation of endowments entails a continuous (smooth) change in the supporting equilibrium price vector.
In other terms, discontinuities of prices are only attributed to catastrophes, that is, when the endowments change crosses a singular economy, as ,e.g.,
$\omega_c$. Otherwise, prices \lq\lq should not jump'' to different \lq\lq records'',
as would occur in $y=(\tilde p,\omega')$.

This view of the smooth selection can be seen mathematically equivalent to finding a smooth function
$f:S\times\Omega(r)\to E(r)$, such that $f(E(r))=E(r)$, or to project homotopically (smoothly) $S\times \Omega(r)$ onto $E(r)$ through a deformation retract.
If we denote by $\pi:S\times \Omega(r)\to\Omega(r)$ the {\em natural projection},
the above construction would imply that, given $(p,\omega)\in E(r)$,  a neighborhood 
$U_{\omega}$ containing $\omega=\pi(p,\omega)$, and  chosen a point
$(p',\omega')$ sufficiently close to $E(r)$, we would have $f(p',\omega')\in\pi^{-1}(U_\omega)$.

The knowledgeable reader in the equilibrium manifold
could think of an alternative, natural way, to tackle this problem using the unknottedness property of $E(r)$ \cite{dmg}, that is, by using a diffeomorphism 
$D:S\times \Omega(r)\to S\times \Omega(r)$ such that $D(E(r))=\{p_0\}\times\Omega(r)$,
where the equilibrium manifold is mapped into an hyperplane. By denoting the projection
$proj:S\times \Omega(r)\to \{p_0\}\times \Omega(r)$, $(p,\omega)\mapsto (p_0,\omega)$,
one could define $f:=D^{-1}\circ proj\circ D:S\times\Omega(r)\to S\times\Omega(r)$, but
unfortunately, this topological approach would not offer any insight on how to explicitly find the function $f$.

The purpose of the present paper instead is to provide a geometric construction to find the function $f$ to explain the smooth selection phenomenon.
To provide an insight, let us denote by  $\pi_T$ the projection $\pi_T:S\times\Omega(r)\to TE(r)$, where $TE(r)$
denotes the tangent bundle of $E(r)$, that is, $\{(p,v)| p\in E(r), v\in T_p(E(r))\}$.

Denote the initial equilibrium 
by $x=(p,\omega)$ and change the endowments to 
$\omega'$.  After this change takes place, the pair
$z=(p,\omega')$ is out of the equilibrium manifold.
Our approach consists of  (1) locally approximating the manifold with its tangent plane at $x$, $T_xE(r)$, (2) projecting $z$ onto $T_xE(r)$, $\pi_T(z)$, and then (3) mapping the vector $v=\pi_T(z)-x$ into the (unique) geodesic connecting 
$x$ to $y$ via the exponential map\footnote{Note that there exists a unique geodesic $\gamma:(-a,a)\to E(r)$, such that $\gamma(0)=x$ and $\gamma'(0)=v$,
for small $a>0$ and $\|v\|$.  Moreover, because by decreasing its velocity, the interval  of definition $(-a,a)$ of the geodesic can be uniformly increased, the exponential map
can be defined as $exp_x(v):=\gamma(1)$. This map can be shown to be a (local) diffeomorphism near $x$. But, in our geometric construction, being $\|v\|$ arbitrary,  $exp_x$ must be defined for all $T_xE(r)$, that is, it must be a global diffeomorphism.} 
$exp_x: TE(r)\to E(r)$,
where $y=(p',\omega')$ represents the new equilibrium
that is consistent with the adjustment process.
This way, the supporting equilibrium price is unambiguously selected and belongs to the same record. This construction enables us to explicitly define the map
 $f$, that is,
 $$f:=exp\circ\pi_T:S\times\Omega(r)\to E(r).$$

The following figure illustrates our approach.

\begin{figure}[h]
\begin{tikzpicture}
  
\draw[<-> ,thick] (7,0) node[below]{$\Omega(r)$} -- (0,0) --
    (0,6) node[left]{$S$};
\draw[thick] (2.5,4.95) to [out=-5,in=120] (6,3);
\draw[fill] (0,4.7) circle [radius=0.05]  node[left]{$p$};   
\draw[fill] (0,3.7) circle [radius=0.05]  node[left]{$p'$};    
\draw[fill] (3.9,4.7) circle [radius=0.05]  node[above]{$x$};
\draw[fill] (5.5,4.7) circle [radius=0.05]  node[above]{$z$};

\draw[dash dot] (0,4.7) --(3.9,4.7);
\draw[dash dot](0,3.7)--(5.5,3.7);


\draw[fill] (5.5,3.7) circle [radius=0.05]  node[right]{$y$};

\draw (3.9,4.7) -- (5.9,4);

\draw[fill] (3.9,0) circle [radius=0.05];
\draw[fill] (5.5,0) circle [radius=0.05];
\draw[dash dot] (5.5,0) --(5.5,3.7);


\draw [-{Stealth[scale=1.5]}] (3.9,4.7) -- (5.35,4.2);
\draw[dash dot] (3.9,0)--(3.9,4.7);

\draw[dash dot] (5.35,4.2)--(5.5,4.7);

\node at (5.1,4.5) {$v$};

\node at (3.9,-0.26) {$\omega$};
\node at (5.5,-0.2) {$\omega'$};

\end{tikzpicture}
\caption{A geometric approach to the equilibrium selection.}
\label{figeo}
\end{figure}
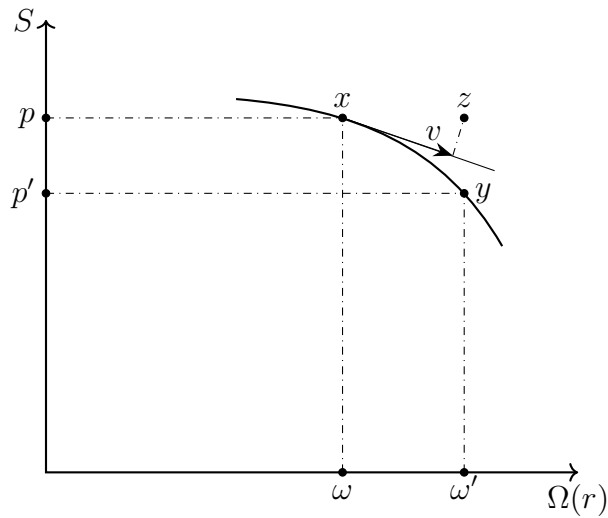
This geometric construction takes into account what occurs outside the manifold
and reflects the principle of optimizing distance,
that can be thought to embody costs of processing information necessary to coordinate agents, thus enabling the system to converge towards the new equilibrium.
Despite its seemingly geometric flavor, it deeply relies on the economic properties
that affect the geometry of the equilibrium manifold. 
Moreover, this optimizing construction is based on the composition of two natural maps, the projection on the linearization of the equilibrium manifold, a method that underlies regression techniques, and the exponential map that associates a tangent vector with a geodesic on the manifold. Furthermore, our approach differs from the contributions in the literature, which focused on introducing uncertainty through randomization over the equilibra. We refer the reader to \cite{pol} and the references therein.

The key property needed to implement our construction is the non-positive curvature $K$, with the induced metric, of the equilibrium manifold. The proof of this property plays a crucial role because it implies, by Hadamard's theorem (see Theorem \ref{had} for details), that the exponential map is a global diffeomorphism. A substantial part of the present paper is devoted to its proof.

There is an interesting connection between curvature of the equilibrium manifold and uniqueness in the literature. Balasko \cite[Theorem 7.3.9]{balib} showed that if there is uniqueness of equilibrium for every endowment profile of the commodity space, then the curvature of the equilibrium manifold is zero. It has been shown \cite{lmcurvuni} that, in the case of two commodities, if the curvature is zero then there is uniqueness of equilibrium.  
Furthermore, following an information-theoretic approach,
\cite{lment} established a connection between  entropy minimization and uniqueness when
the equilibrium manifold is a minimal stable submanifold of its ambient space, a property that can be expressed through a minimality condition in terms of the vanishing of the mean curvature.
Moreover, \cite{lmcurvuni} conjectured that the equivalence between zero curvature property and uniqueness holds for an arbitrary number of goods. As a by-product of our main result, in the present paper we prove (see Corollary \ref{coruni})
this equivalence in the case of an arbitrary number of goods and two consumers, thus extending the previous result in the direction of the conjecture.

The rest of this paper is organized as follows. In Section \ref{setting}
we recall the properties of the equilibrium manifold relevant for our purposes.
In Section \ref{tools} we introduce the main concepts of differential geometry
used in this paper. In Section \ref{main} we prove (see Theorem \ref{teorcurv}) that the equilibrium manifold has non-positive sectional curvature, a property that legitimates our geometric approach to the equilibrium selection.
Moreover, Corollary \ref{coruni} establishes the equivalence between zero-curvature and uniqueness of equilibrium. Finally, a mathematical appendix contains all the tedious computations.

\section{The economic setting}\label{setting}

We consider a smooth pure exchange economy with $L$ goods and $M$ consumers. 
The equilibrium manifold $E$ is defined as the set
of pairs $(p,\omega)$ such that the excess demand function is zero, where
$p$ belongs to the set of normalized prices $S=\{p=(p_1,\ldots,p_L)\in \R^L| p_l>0,l=1,\ldots,L,p_L=1\}\cong \mathbb{R}^{L-1}$ and $\omega$ belongs to the space of endowments $\Omega=\R^{LM}$.
The equilibrium manifold enjoys very nice geometric properties, being a smooth submanifold of $S\times\Omega$ globally diffeomorphic to $\R^{LM}$ \citep[Lemma 3.2.1]{balib}. 

If total resources $r\in\R^L$ are fixed,  the equilibrium manifold, denoted by $E(r)$, is a submanifold of $S\times\Omega(r)$ globally diffeomorphic to $\R^{L(M-1)}$ \citep[Corollary 5.2.5]{balib}, that is, $E(r)\cong B(r)\times \R^{(L-1)(M-1)}$, where $B(r)$ denotes the {\em price-income equilibria}, a submanifold of $S\times \R^M$ diffeomorphic to $\R^{M-1}$ \citep[Corollary 5.2.4]{balib}.
If we define this diffeomorphism as

\begin{equation}\label{parB}
\phi:\R^{M-1}\to B(r),
\end{equation}

$$t=(t_1,\ldots,t_{M-1})\mapsto  (p(t),w_1(t),\ldots,w_{M-1}(t)),$$

where $w_i$ denotes consumer $i$'s income, a parametrization of $E(r)$ 
(see formulas (6), (7) and (10) in \cite{lmcurvuni}) is given by 

\begin{equation}\label{parE}
\Phi:\R^{L(M-1)}\to E(r),
\end{equation}

$$(t,\bar\omega_1,\ldots,\bar\omega_{M-1})\mapsto (p(t),\bar\omega_1,w_1(t)-p(t)\bar\omega_1,\ldots
\bar\omega_{M-1},w_{M-1}(t)-p(t)\bar\omega_{M-1}),$$
where $\bar\omega_i$ denotes the first $L-1$ components of $\omega_i$,
consumer $i$'s endowments vector.

\section{Geometric tools}\label{tools}

We refer the reader to \cite{do} for a deeper understanding of the concepts of
differential geometry used in this paper. In this section we introduce the main tools.

Let $M$ be  a submanifold of dimension $m$ in its ambient space  $(\mathbb R^n,g_{euclid})$. This induces on $M$ a metric in a natural way. In particular, if   
\[
\begin{split}
\phi :  &\: \mathbb{R}^m\longrightarrow M\subset \mathbb R^n\\
&(x_1,\dots,x_m)\mapsto (\phi_1,\dots\phi_n)
\end{split}\] 
is a parametrization of $M$, the vector fields  
$X_1=(\frac{\partial \phi_1}{\partial x_1},\dots,\frac{\partial \phi_1}{\partial x_m})$, $X_2=(\frac{\partial \phi_2}{\partial x_1},\dots,\frac{\partial \phi_2}{\partial x_m})$ and  $X_n=(\frac{\partial \phi_n}{\partial x_1},\dots,\frac{\partial \phi_n}{\partial x_m})$
form a basis
$\{X_1,\dots, X_n\}$  of vector fields of $T_pM$ for $p\in M$. The induced metric is given by $$ds^2=\displaystyle\sum_{i,j=1}^{m}g_{ij}dx_idx_j$$ where $g_{ij}=\langle X_i,X_j\rangle_{g_{euclid}}$. \\
Let us denote by $\Chi (M)$ the set of all vector fields of class $C^{\infty}$ on $M$. Then there exists an affine connection
 \[\begin{split}
 \nabla: \Chi(M)\times& \Chi(M)\rightarrow \Chi(M)\\
 (X,&Y)\mapsto \nabla_XY
 \end{split}\] that satisfies the following properties:
 \begin{itemize}
 \item $\nabla_{fX+gY}=f\nabla _XZ+g\nabla_YZ$
 \item $\nabla_X(Y+Z)=\nabla_XY+\nabla_XZ$
 \item $\nabla_X(f Y)=f\nabla_X Y+X(f)Y,$
 \end{itemize} 
 with $f, g$ real-valued functions of class $C^{\infty}$ on $M$.

\begin{theorem}[Levi-Civita (\cite{do}, p.55)] Given a Riemannian manifold $M$, there exists a unique affine connection $\nabla$ on $M$ satisfying the conditions:
\begin{itemize}
\item $\nabla $ is symmetric
\item $\nabla$ is compatible \footnote{A connection $\nabla$ on a Riemannian manifold $M$ is compatible with the metric if and only if $X\langle Y,Z\rangle=\langle \nabla_XY,Z\rangle+ \langle Y,\nabla_XZ\rangle$ for $X, Y, Z\in \Chi (M)$} with the Riemannian  metric.
\end{itemize}
\end{theorem}

In particular, the Levi-Civita connection can be written, in a coordinate system $(U,x)$, as 
$$\displaystyle\nabla_{\frac{\partial}{\partial x_i}}\frac{\partial}{\partial x_j}=\displaystyle\sum_{k=1}^{m}\Gamma_{ij}^k\frac{\partial}{\partial x_k}.$$ 
where the coefficients $\Gamma_{ij}^k$ are called the {\em Christoffell symbols} and can be computed with the following formula
$$\Gamma_{ij}^k=\frac{1}{2}\displaystyle\sum_{h=1}^{m}g^{hk}\left( \frac{\partial g_{jh}}{\partial x_i} + \frac{\partial g_{hi}}{\partial x_j}-\frac{\partial g_{ij}}{\partial x_h}\right),$$
where $g^{ij}$ is the inverse matrix of $g_{ij}=\langle X_i,X_j\rangle$.\\

The {\em curvature tensor} intuitively measures the deviation of a manifold from being locally Euclidean.

\begin{definition}
The curvature tensor $R$ of a Riemannian manifold $M$ is a correspondence that associates to every pair $X,Y\in \Chi(M)$ a mapping 
 \[\begin{split}
 R(X,Y): &\: \Chi(M)\rightarrow \Chi(M)\\
 &Z\mapsto R(X,Y)Z,
 \end{split}\]
 where  $$R(X,Y)Z=\nabla_Y\nabla_XZ-\nabla_X\nabla_YZ+\nabla_{[X,Y]}Z$$
\end{definition} 
If $M=\mathbb R^n$, then $R(X,Y)Z=0$ for all $X,Y,Z\in \Chi(\mathbb R^n)$.
It is convenient to express this curvature in a coordinate system $(U,x)$ based at the point $p\in M$.  We have $$R(X_i,X_j)X_k= \displaystyle \sum _l R^l_{ijk} X_l,$$
 where the coefficients $R_{ijk}^l$ can be expressed in terms of $\Gamma^k_{ij}$  
  $$R_{ijk}^s=\displaystyle\sum_{l=1}^{m}\Gamma^l_{ik}\Gamma_{jl}^s-\displaystyle\sum_{l=1}^{m}\Gamma^l_{jk}\Gamma_{il}^s+ \frac{\partial \Gamma_{ik}^s}{\partial x_j} - \frac{\partial \Gamma^s_{jk}}{\partial x_i}.$$
Moreover, we have
 $$\langle R(X_i,X_j)X_k,X_s\rangle=\displaystyle \sum _l R^l_{ijk}g_{ls}.$$
We introduce the {\em sectional curvature} of $M$, which generalizes the Gaussian
curvature of surfaces. 
Let $M$ be a Riemannian $n-$manifold and let $p\in M$. If $\Pi$ is any $2$-dimensional subspace of $T_pM$, and $U$ is a neighborhood of zero on which $\exp_p$ is a diffeomorphism, then $S_{\Pi}:=\exp_p(\Pi \cap U)$ is a $2$-dimensional submanifold of $M$ containing $p$. Then the sectional curvature of $M$ associated with $\Pi$ is the Gaussian Curvature of $S_{\Pi}$. If $\{X, Y\}$ is any basis for $\Pi$, we indicate the sectional curvature as $K(X,Y)$ and we have
\begin{definition}[]If $\{X, Y\}$ is any basis for a $2$-plane  $\Pi\in T_p M$, then
$$K(X,Y)=\frac{\langle R(X,Y)X,Y\rangle}{\vert X\vert ^2\vert Y\vert^2-\langle X, Y\rangle^2}.$$
\end{definition}

A well-known theorem by Hadamard establishes an important connection 
between local and global properties of a differential manifold.

\begin{theorem}[Hadamard (\cite{do}, p. 149]\label{had}
Let $M$ be a complete Riemannian manifold, simply connected
with sectional curvature $K(p,\sigma)\leq 0$, for all $p\in M$ and for
all $\sigma\in T_p(M)$. Then $M$ is diffeomorphic to $\R^n$,
$n=dim M$; more precisely, $exp_p: T_pM\to M$ is a diffeomorphism.
\end{theorem}

\section{Main results}\label{main}

Theorem \ref{had} represents a key result for our construction, 
because we need to associate, through the exponential map, a vector belonging to the tangent space of $E(r)$ to a geodesic. Hence, it is crucial for our purposes that
the exponential map is a global diffeomorphism. 
Since $E(r)$ is diffeomorphic to an Euclidean space, and hence complete and simply connected, all we need to prove is that its sectional curvature is non-positive.
In the following theorem, we prove this property for the case $M=2$.

\begin{theorem}\label{teorcurv}
Let $M=2$. Then the equilibrium manifold $E(r)$
has non-positive sectional curvature.
\end{theorem}

\begin{proof}
If $M=2$, the manifold $B(r)$ is diffeomorphic to $\R$ through the map (see \eqref{parB}
above)

\begin{equation*} 
\begin{split} 
 \phi:\: \mathbb{R} &  {\longrightarrow} B(r)\subset S\times \mathbb{R}^{M-1}=S\times \mathbb{R}\\
       t & \longmapsto  (p(t),w(t)),\\ 
\end{split} 
\end{equation*}

and $E(r)$ is a submanifold of dimension $L$ in a space of dimension $2L-1$.
By setting $\alpha_i:=\omega^i$, a parametrization of $E(r)$ is given by (see \eqref{parE}
above)
\begin{equation*}  
\begin{split} 
 \Phi:\: \mathbb{R}^L &  \overset{\:\:\:\:\:\:\:\:\:\:\:\:\:\:\:\:\:\:\:\:\:\:\:\:\:\:\:}{\longrightarrow} E(r),\\
        (t,\alpha_1,\dots,\alpha_{L-1} )& \longmapsto  (p_1(t),\dots,p_{L-1}(t),\alpha_1,\dots,\alpha_{L-1},w(t)-p_1(t)\alpha_1-\dots-p_{L-1}(t)\alpha_{L-1}).\\ 
\end{split} 
\end{equation*}

Consider a basis of a vector field of $T_xE(r)$ given by
$$\Phi_0=\left(\frac{\partial p_1}{\partial t},\dots,\frac{\partial p_{L-1}}{\partial t} ,0,\dots,0,\frac{\partial w}{\partial t}-\frac{\partial p_1}{\partial t}\alpha_1-\dots-\frac{\partial p_{L-1}}{\partial t}\alpha_{L-1}\right)=(p_1',\dots,p_{L-1}',0,\dots,0,A)$$
$$\Phi_{1}=\left(0,\dots, 0,1,0,0,-p_1(t)\right)$$
$$\Phi_{i}=\left(0,\dots, 0,0,1,0,-p_i(t)\right)$$
where $1$ is in the $L-1+i$ position and where
$A=\frac{\partial w}{\partial t}-\frac{\partial p_1}{\partial t} \alpha_1-\dots-\frac{\partial p_{L-1}}{\partial t}\alpha_{L-1}$
or, more compactly,  \footnote {$A=w'-\langle p',\alpha \rangle$ with $p'=(p_1',\dots,p_{L-1}')$ and $\alpha=(\alpha_1,\dots,\alpha_{L-1})$.}
$$A= w'-p'_1 x_1-\dots-p'_{L-1}x_{L-1}.$$
Clearly, we have $A(x_0,x_1,\dots,x_{L-1})$ and
$\frac{\partial }{\partial x_i}A=-p_i'$, for all $i\neq 0$.\\
We set $x_0=t,\:\: x_1=\alpha_1,\:\:x_2=\alpha_2,\:\:x_{L-1}=\alpha_{L-1}$ and $X_i=\Phi_i$.\\
The induced metric on $E(r)$ is given by
$$ds^2=\displaystyle\sum_{i,j=0}^{L-1}g_{ij}dx_idx_j.$$
Set 
$B(x_0)=\left(\frac{\partial p_1}{\partial x_0}\right)^2+\dots+\left(\frac{\partial p_{L-1}}{\partial x_0}\right)^2$ or, more compactly,
$$B(x_0)=(p_1')^2+\dots +(p_{L-1}')^2.$$
An easy calculation gives
\[\begin{split}
&\langle \Phi_0,\Phi_0\rangle=B+A^2\\
&\langle \Phi_{0},\Phi_{i}\rangle=-p_iA\\
&\langle \Phi_{i},\Phi_{i}\rangle=1+p_i^2\\
&\langle \Phi_{i},\Phi_{j}\rangle=p_ip_j\\
\end{split}
\]

\[ 
g_{ij}=\left( 
   \begin{array}{cccccc} 
     B+A^2 & \;\;-p_1 A& \dots &\:\:-p_iA&\dots & \:\:-p_{L-1}A \\ 
\rule[0.5cm]{0mm}{1cm}
    -p_1A& 1+p_1^2 &\dots& p_1p_i&\dots&p_1p_{L-1}\\
 \rule[0.5cm]{0mm}{1cm}    
    \dots& \dots&\dots& \dots&\dots&\dots\\   
\rule[0.5cm]{0mm}{1cm}    
    -p_iA& p_1p_i&\dots& 1+p_i^2&\dots&p_ip_{L-1}\\
    \rule[0.5cm]{0mm}{1cm}    
    \dots& \dots&\dots& \dots&\dots&\dots\\
\rule[0.5cm]{0mm}{1cm}    
    -p_{L-1}A& p_1p_{L-1}&\dots& p_ip_{L-1}&\dots&1+p_{L-1}^2
   \end{array} 
\right). 
\] 
Setting
$$g:=(1+p_1^2+\dots+p_{L-1}^2)B+A^2=\Vert p\Vert ^2B+A^2,$$
where $$ \fbox{$\Vert p\Vert ^2=(1+p_1^2+\dots+p_{L-1}^2)$}, $$
the inverse matrix $g^{ij}$ can be written as

{\tiny\[
\frac{1}{ g }\\
\left( 
   \begin{array}{cccccc} 
    1+p_1^2+\dots+p_{L-1}^2 & \;\;p_1A&\dots& \:\:p_iA&\dots& \:\:p_{L-1}A \\ 
\rule[0.5cm]{0mm}{1cm}
    p_1A& (1+p_2^2+\dots+p_{L-1}^2)B+A^2&\dots& -p_1p_iB&\dots&-p_1p_{L-1}B\\ 
 \rule[0.5cm]{0mm}{1cm}    
    \dots& \dots&\dots& \dots&\dots&\dots\\  
\rule[0.5cm]{0mm}{1cm}    
    p_iA& -p_1p_iB&\dots & (1+p_1^2+\dots+p_{i-1}^2+p_{i+1}^2+\dots+p_{L-1}^2)B+A^2&\dots&-p_ip_{L-1}B\\
    \rule[0.5cm]{0mm}{1cm}    
  \dots& \dots&\dots& \dots&\dots&\dots\\  
 \rule[0.5cm]{0mm}{1cm}     
   p_{L-1}A& -p_1p_{L-1}B&\dots&-p_{i}p_{L-1}B&\dots& (1+p_1^2+\dots+p_{L-2}^2)B+A^2
   \end{array} 
\right) 
\] 
}

\noindent To compute the Christofell symbols
$$\Gamma_{ij}^k=\frac{1}{2}\displaystyle\sum_{h=1}^{L-1}g^{hk}\left( \frac{\partial g_{jh}}{\partial x_i} + \frac{\partial g_{hi}}{\partial x_j}-\frac{\partial g_{ij}}{\partial x_h}\right),$$
observe that the entries $g_{ij}$ with $i$ and $j$ both different from $0$ do not depend on the $x_i$, so their derivatives with respect to $x_k$, with $k\neq 0$, are zero.\\
Moreover, $\frac{\partial}{\partial x_i}g_{0i}=-p_{i}(-p'_{i})=p_{i}p_{i}'$, while $\frac{\partial}{\partial x_0}g_{ii}=2p_{i}p'_{i}$.
Hence all Christofell symbols with subscript different from $0$ vanish.
The others symbols can be obtained through a long but straightforward calculation (see Appendix \ref{Chris}). Observe that $\frac{\partial B}{\partial x_0}=2(p_1'p_1''+p_2'p_2''+\dots+p_{L-1}'p_{L-1}'')$, so for convenience we set  $$C:=p_1'p_1''+p_2'p_2''+\dots+p_{L-1}'p_{L-1}''$$ and 
 $$A':= w''- p_1''x_1-\dots- p_{L-1}''x_{L-1} $$ 
The Christoffel symbols are
\[
\begin{split}
\Gamma_{00}^0=&\frac{[\Vert p\Vert^2C+AA']}{\Vert p \Vert ^2B+A^2},\\
\Gamma_{00}^k=&\frac{p_{k}\left[AC-A'B\right]}{\Vert p \Vert ^2B+A^2},\\
\Gamma_{0j}^0=&\frac{-p'_{j}A}{\Vert p \Vert ^2B+A^2},\\
\Gamma_{0j}^k=&\frac{p'_{j}p_{k}B}{\Vert p \Vert ^2B+A^2},\\
\Gamma^0_{ij}=&\Gamma^0_{ii}=0,\\
\Gamma^k_{ij}=&\Gamma^k_{ii}=0. \end{split}\]

To determine the sectional curvature 

$$K(X,Y)=\frac{\langle R(X,Y)X,Y\rangle}{\vert X\wedge Y\vert},$$
we need to calculate the coefficients of the tensor $R$ using 
\begin{equation}\label{RC}
R_{ijk}^s=\displaystyle\sum_{m=0}^{l-1}\Gamma^m_{ik}\Gamma_{jm}^s-\displaystyle\sum_{m=0}^{l-1}\Gamma^m_{jk}\Gamma_{im}^s+ \frac{\partial \Gamma_{ik}^s}{\partial x_j} - \frac{\partial \Gamma^s_{jk}}{\partial x_i}.\end{equation}

In particular, we have that $K(X_i,X_j)=0$ for all $i,j\neq 0$.

Since $$K(X_0,X_i)=\frac{\langle R(X_0,X_i)X_0,X_i\rangle}{\vert X_0\wedge X_i\vert}$$
we use (\ref{RC}) to compute $\langle R(X_0,X_i)X_0,X_i\rangle=\displaystyle \sum _s R^s_{0i0}g_{si}$.

After a long but straightforward calculation (see Appendix \ref{Rcoeff}),
we obtain \[\begin{split}
\langle R(X_0,X_i)X_0,X_i\rangle=&\displaystyle \sum _s R^{s}_{0i0}g_{si}=R^0_{0i0}g_{0i}+R^1_{0i1}g_{1i}+\dots+R^{l-1}_{0i0}g_{(l-1)i}=\\
=&-\frac{(p_{i}')^2B}{(\det g)}\leq 0.
\end{split} \]

\end{proof}

As a by-product, the following corollary complements \cite[Theorem 5.1] {lmcurvuni}, where they show that for $L=2$ the zero-curvature condition on the equilibrium 
manifold is equivalent to the global uniqueness of the equilibrium price.

\begin{corollary}\label{coruni}
Let $M= 2$. A necessary and sufficient condition for a unique equilibrium price is that the curvature of E(r) is zero.
\end{corollary}

\begin{proof}
By \cite[Theorem 7.3.9]{balib},
if for every $\omega\in\Omega(r)$ there is an unique equilibrium,
then the price  $p$ associated to $\omega$ does not depend on $\omega$, that is, $E(r)$ is an hyperplane and hence its curvature is zero. Conversely, by Theorem \ref{teorcurv},
if its sectional curvature is zero, then $p'_i=0$, that is, the price is constant and hence unique.
\end{proof}

\appendix

\section{Christoffel Symbols}\label{Chris}

{\tiny

{\bf Type: $\Gamma_{ii}^k$ where $i,k\neq 0$.}
\[\begin{split}
\Gamma_{ii}^k=&\frac{1}{2}\displaystyle\sum_{h=0}^{L-1}g^{hk}\left( \frac{\partial g_{ih}}{\partial x_i} + \frac{\partial g_{hi}}{\partial x_i}-\frac{\partial g_{ii}}{\partial x_h}\right)=\\
=&\frac{1}{2}\left[g^{0k}\left( \frac{\partial g_{i0}}{\partial x_i} + \frac{\partial g_{0i}}{\partial x_i}-\frac{\partial g_{ii}}{\partial x_0}\right)+\dots+g^{(L-1)k}\left( \frac{\partial g_{i(L-1)}}{\partial x_i} + \frac{\partial g_{(L-1)i}}{\partial x_i}- \frac{\partial g_{ii}}{\partial x_{L-1}} \right)\right]\end{split}\]

\[\begin{split}
\Gamma_{ii}^k=&
\frac{1}{2}\left[g^{0k}\left( \frac{\partial g_{i0}}{\partial x_i} + \frac{\partial g_{0i}}{\partial x_i}-\frac{\partial g_{ii}}{\partial x_0}\right)+
g^{1k}\left( \cancel {\frac{\partial g_{i1}}{\partial x_i}} +\cancel{ \frac{\partial g_{1i}}{\partial x_i}}-\cancel{\frac{\partial g_{ii}}{\partial x_1}}\right)+g^{2k}\left(\cancel{ \frac{\partial g_{i2}}{\partial x_i} }+ \cancel{\frac{\partial g_{2i}}{\partial x_i}}-\cancel{ \frac{\partial g_{ii}}{\partial x_2}} \right)+\dots \right]=\\
=&\frac{1}{2}\left[g^{0k}\left(2 \frac{\partial g_{i0}}{\partial x_i} -\frac{\partial g_{ii}}{\partial x_0}\right)\right]=\frac{1}{2}\left[g^{0k}\left(2 p_{i}p_{i}'-2p_{i}p_{i}'\right)\right]=0
\end{split}
\]
{\bf Type: $\Gamma_{ij}^k$ where $i,j,k\neq 0$.}\\ 
Using $\frac{\partial}{\partial x_i}g_{0j}=p_{j}p'_{i}$ and  $\frac{\partial}{\partial x_0}g_{ij}=\frac{\partial}{\partial x_0}(p_{i}p_{j})=p'_{i}p_{j}+p_{i}p'_{j}$.

\[\begin{split}
\Gamma_{ij}^k=&\frac{1}{2}\displaystyle\sum_{h=0}^{L-1}g^{hk}\left( \frac{\partial g_{jh}}{\partial x_i} + \frac{\partial g_{hi}}{\partial x_j}-\frac{\partial g_{ij}}{\partial x_h}\right)=\\
=&\frac{1}{2}\left[g^{0k}\left( \frac{\partial g_{j0}}{\partial x_i} + \frac{\partial g_{0i}}{\partial x_j}-\frac{\partial g_{ij}}{\partial x_0}\right)+
g^{1k}\left( \frac{\partial g_{j1}}{\partial x_i} + \frac{\partial g_{1i}}{\partial x_j}-\frac{\partial g_{ij}}{\partial x_1}\right)+g^{2k}\left( \frac{\partial g_{j2}}{\partial x_i} + \frac{\partial g_{2i}}{\partial x_j}- \frac{\partial g_{ij}}{\partial x_2} \right)+\dots \right]\end{split}\]

\[\begin{split}
\Gamma_{ij}^k=&
\frac{1}{2}\left[g^{0k}\left( \frac{\partial g_{j0}}{\partial x_i} + \frac{\partial g_{0i}}{\partial x_j}-\frac{\partial g_{ij}}{\partial x_0}\right)+
g^{1k}\left( \cancel {\frac{\partial g_{j1}}{\partial x_i}} +\cancel{ \frac{\partial g_{1i}}{\partial x_j}}-\cancel{\frac{\partial g_{ij}}{\partial x_1}}\right)+g^{2k}\left(\cancel{ \frac{\partial g_{j2}}{\partial x_i} }+ \cancel{\frac{\partial g_{2i}}{\partial x_j}}-\cancel{ \frac{\partial g_{ij}}{\partial x_2}} \right)+\dots \right]=\\
=&\frac{1}{2}\left[g^{0k}\left(\frac{\partial g_{j0}}{\partial x_i} + \frac{\partial g_{0i}}{\partial x_j}-\frac{\partial g_{ij}}{\partial x_0}\right)\right]=\frac{1}{2}\left[g^{0k}\left(p_{j}p'_{i}+p_{i}p'_{j}-(p'_{i}p_{j}+p_{i}p'_{j})\right)\right]=0\\
\end{split}
\]
{\bf Type: $\Gamma_{0j}^0$ where $j\neq 0$.}\\ 
Using $\frac{\partial}{\partial x_j}g_{00}=-2Ap'_{j}$. \\
\[\begin{split}
\Gamma_{0j}^0=&\frac{1}{2}\displaystyle\sum_{h=0}^{L-1}g^{h0}\left( \frac{\partial g_{jh}}{\partial x_0} + \frac{\partial g_{h0}}{\partial x_j}-\frac{\partial g_{0j}}{\partial x_h}\right)=\\
=&\frac{1}{2}\left[g^{00}\left( \frac{\partial g_{j0}}{\partial x_0} + \frac{\partial g_{00}}{\partial x_j}-\frac{\partial g_{0j}}{\partial x_0}\right)+
g^{10}\left( \frac{\partial g_{j1}}{\partial x_0} + \frac{\partial g_{10}}{\partial x_j}-\frac{\partial g_{0j}}{\partial x_1}\right)+g^{20}\left( \frac{\partial g_{j2}}{\partial x_0} + \frac{\partial g_{20}}{\partial x_j}- \frac{\partial g_{0j}}{\partial x_2} \right)+\dots \right]
\end{split}\]

\[\begin{split}
\Gamma_{0j}^0=&
\frac{1}{2}\left[g^{00}\left( \cancel{\frac{\partial g_{j0}}{\partial x_0} }+ \frac{\partial g_{00}}{\partial x_j}-\cancel{\frac{\partial g_{0j}}{\partial x_0}}\right)+
g^{10}\left( \frac{\partial g_{j1}}{\partial x_0} + \frac{\partial g_{10}}{\partial x_j}-\frac{\partial g_{0j}}{\partial x_1}\right)+g^{20}\left( \frac{\partial g_{j2}}{\partial x_0} + \frac{\partial g_{20}}{\partial x_j}- \frac{\partial g_{0j}}{\partial x_2} \right)+\dots \right]=\\
=&\frac{1}{2}\left[g^{00}\left( \frac{\partial g_{00}}{\partial x_j}\right)+
g^{10}\left( \frac{\partial g_{j1}}{\partial x_0} + \frac{\partial g_{10}}{\partial x_j}-\frac{\partial g_{0j}}{\partial x_1}\right)+g^{20}\left( \frac{\partial g_{j2}}{\partial x_0} + \frac{\partial g_{20}}{\partial x_j}- \frac{\partial g_{0j}}{\partial x_2} \right)+\dots \right]=\\
=&\frac{1}{2}\left[g^{00}\left( -2Ap'_{j}\right)+
g^{10}\left( p_1'p_{j}+p_1p'_{j} +p_1p'_{j}-p_{j}p'_1\right)+g^{20}\left( p_2'p_{j}+p_2p'_{j} +p_2p'_{j}-p_{j}p'_2 \right)+\dots \right]=\\
=&\frac{1}{2}\left[g^{00}\left( -2Ap'_{j}\right)+
g^{10}\left( 2p_1p'_{j} \right)+g^{20}\left( 2p_2p'_{j}  \right)+\dots +g^{(L-1)0}\left( 2p_{n-1}p'_{j}  \right)\right]=\\
=&\frac{p'_{j}}{\det g }\left[(1+p_1^2+\dots+p_{L-1}^2)\left( -A\right)+
p_1A\left( p_1 \right)+p_2A\left( p_2  \right)+\dots +p_{L-1}A\left( p_{L-1} \right)\right]=\frac{-p'_{j}A}{\det g }\\\
\end{split}
\]
{\bf Type: $\Gamma_{0j}^k$ where $j,k\neq 0$.}\\ 
\[\begin{split}
\Gamma_{0j}^k&=\frac{1}{2}\displaystyle\sum_{h=0}^{L-1}g^{hk}\left( \frac{\partial g_{jh}}{\partial x_0} + \frac{\partial g_{h0}}{\partial x_j}-\frac{\partial g_{0j}}{\partial x_h}\right)=\\
&=\frac{1}{2}\left[g^{0k}\left( \frac{\partial g_{j0}}{\partial x_0} + \frac{\partial g_{00}}{\partial x_j}-\frac{\partial g_{0j}}{\partial x_0}\right)+
g^{1k}\left( \frac{\partial g_{j1}}{\partial x_0} + \frac{\partial g_{10}}{\partial x_j}-\frac{\partial g_{0j}}{\partial x_1}\right)+g^{2k}\left( \frac{\partial g_{j2}}{\partial x_0} + \frac{\partial g_{20}}{\partial x_j}- \frac{\partial g_{0j}}{\partial x_2} \right)+\dots \right]\end{split}\]

\[\begin{split}
\Gamma_{0j}^k&=
\frac{1}{2}\left[g^{0k}\left( \cancel{\frac{\partial g_{j0}}{\partial x_0} }+ \frac{\partial g_{00}}{\partial x_j}-\cancel{\frac{\partial g_{0j}}{\partial x_0}}\right)+
g^{1k}\left( \frac{\partial g_{j1}}{\partial x_0} + \frac{\partial g_{10}}{\partial x_j}-\frac{\partial g_{0j}}{\partial x_1}\right)+g^{2k}\left( \frac{\partial g_{j2}}{\partial x_0} + \frac{\partial g_{20}}{\partial x_j}- \frac{\partial g_{0j}}{\partial x_2} \right)+\dots \right]=\\
&=\frac{1}{2}\left[g^{0k}\left( \frac{\partial g_{00}}{\partial x_j}\right)+
g^{1k}\left( \frac{\partial g_{j1}}{\partial x_0} + \frac{\partial g_{10}}{\partial x_j}-\frac{\partial g_{0j}}{\partial x_1}\right)+g^{2k}\left( \frac{\partial g_{j2}}{\partial x_0} + \frac{\partial g_{20}}{\partial x_j}- \frac{\partial g_{0j}}{\partial x_2} \right)+\dots \right]=\\
&=\frac{1}{2}\left[g^{0k}\left( -2Ap'_{j}\right)+
g^{1k}\left( p_1'p_{j}+p_1p'_{j} +p_1p'_{j}-p_{j}p'_1\right)+g^{2k}\left( p_2'p_{j}+p_2p'_{j} +p_2p'_{j}-p_{j}p'_2 \right)+\dots \right]=\\
&=\frac{1}{2}\left[g^{0k}\left( -2Ap'_{j}\right)+
g^{1k}\left( 2p_1p'_{j} \right)+g^{2k}\left( 2p_2p'_{j}  \right)+\dots +g^{(L-1)k}\left( 2p_{L-1}p'_{j}  \right)\right]=\\
&=\frac{p'_{j}}{\det g }\left[(p_{k}A\left( -A\right)-p_1p_{k}B
\left( p_1 \right)-p_2p_{k}B
\left( p_2 \right)+\dots +\left[(1+p_1^2+\dots+p_{k-1}^2+p_{k+1}^2+\dots+p_{L-1}^2)B+A^2\right](p_{k})-\right.  \\
&-\left. \dots-p_{k}p_{L-1}B\left( p_{L-1} \right)\right]\\
&=\frac{p'_{j}p_{k}}{\det g }\left[-A^2-p_1^2B
-p_2^2B+\dots +\left[(1+p_1^2+\dots+p_{k-1}^2+p_{k+1}^2+\dots+p_{L-1}^2)B+A^2+\dots-p_{L-1}^2B\right)\right]\\
&=\frac{p'_{j}p_{k}}{\det g }B\\
\end{split}
\]
{\bf Type: $\Gamma_{00}^0$.}\\  
Using $\frac{\partial}{\partial x_0}g_{0i}=-p_{i}'A-p_{i}A'$ and $\frac{\partial}{\partial x_0}g_{00}=-2C+2AA'$  
\[\begin{split}
\Gamma_{00}^0&=\frac{1}{2}\displaystyle\sum_{h=0}^{L-1}g^{h0}\left( \frac{\partial g_{0h}}{\partial x_0} + \frac{\partial g_{h0}}{\partial x_0}-\frac{\partial g_{00}}{\partial x_h}\right)=\\
&=\frac{1}{2}\left[g^{00}\left( \frac{\partial g_{00}}{\partial x_0} + \frac{\partial g_{00}}{\partial x_0}-\frac{\partial g_{00}}{\partial x_0}\right)+
g^{10}\left( \frac{\partial g_{01}}{\partial x_0} + \frac{\partial g_{10}}{\partial x_0}-\frac{\partial g_{00}}{\partial x_1}\right)+g^{20}\left( \frac{\partial g_{02}}{\partial x_0} + \frac{\partial g_{20}}{\partial x_0}- \frac{\partial g_{00}}{\partial x_2} \right)+\dots \right]\end{split}\]

\[\begin{split}
\Gamma_{00}^0=&
\frac{1}{2}\left[g^{00}\left( \cancel{\frac{\partial g_{00}}{\partial x_0} }+ \frac{\partial g_{00}}{\partial x_0}-\cancel{\frac{\partial g_{00}}{\partial x_0}}\right)+
g^{10}\left( \frac{\partial g_{01}}{\partial x_0} + \frac{\partial g_{10}}{\partial x_0}-\frac{\partial g_{00}}{\partial x_1}\right)+g^{20}\left( \frac{\partial g_{02}}{\partial x_0} + \frac{\partial g_{20}}{\partial x_0}- \frac{\partial g_{00}}{\partial x_2} \right)+\dots \right]=\\
=&\frac{1}{2}\left[g^{00}\left( \frac{\partial g_{00}}{\partial x_0}\right)+
g^{10}\left( 2\frac{\partial g_{01}}{\partial x_0} -\frac{\partial g_{00}}{\partial x_1}\right)+g^{20}\left( 2\frac{\partial g_{02}}{\partial x_0} - \frac{\partial g_{00}}{\partial x_2} \right)+\dots \right]=\\
=&\frac{1}{2}\left[g^{00}\left( 2C+2AA'\right)+
g^{10}\left( -2p'_1A-2p_1A'+2Ap_1'\right)+g^{20}\left( -2p_2'A-2p_2A'+2Ap_2'\right)+\dots \right]=\\
=&\frac{1}{2}\left[g^{00}\left( 2C+2AA'\right)+
g^{10}\left(\cancel{ -2p'_1A}-2p_1A'+\cancel{2Ap_1'}\right)+g^{20}\left( \cancel{-2p_2'A}-2p_2A'+\cancel{2Ap_2'}\right)+\dots \right]=\\
=&\frac{1}{\det g }\left[(1+p_1^2+\dots+p_{L-1}^2)\left( C+AA'\right)+
p_1A\left( -A'p_1 \right)+p_2A\left(-p_2A'  \right)+\dots +p_{L-1}A\left( -p_{L-1}A' \right)\right]=\\
=&\frac{1}{\det g }\left[(1+p_1^2+\dots+p_{L-1}^2)C+AA'\right]=\frac{1}{\det g }\left[\Vert p\Vert ^2C+AA'\right]\\
\end{split}
\]
{\bf Type: $\Gamma_{00}^k$ where $k\neq 0$.}\\  
\[\begin{split}
\Gamma_{00}^k&=\frac{1}{2}\displaystyle\sum_{h=0}^{L-1}g^{hk}\left( \frac{\partial g_{0h}}{\partial x_0} + \frac{\partial g_{h0}}{\partial x_0}-\frac{\partial g_{00}}{\partial x_h}\right)=\\
&=\frac{1}{2}\left[g^{0k}\left( \frac{\partial g_{00}}{\partial x_0} + \frac{\partial g_{00}}{\partial x_0}-\frac{\partial g_{00}}{\partial x_0}\right)+
g^{1k}\left( \frac{\partial g_{01}}{\partial x_0} + \frac{\partial g_{10}}{\partial x_0}-\frac{\partial g_{00}}{\partial x_1}\right)+g^{2k}\left( \frac{\partial g_{02}}{\partial x_0} + \frac{\partial g_{20}}{\partial x_0}- \frac{\partial g_{00}}{\partial x_2} \right)+\dots \right]
\end{split}\]

\[\begin{split}
\Gamma_{00}^k=&
\frac{1}{2}\left[g^{0k}\left( \cancel{\frac{\partial g_{00}}{\partial x_0} }+ \frac{\partial g_{00}}{\partial x_0}-\cancel{\frac{\partial g_{00}}{\partial x_0}}\right)+
g^{1k}\left( \frac{\partial g_{01}}{\partial x_0} + \frac{\partial g_{10}}{\partial x_0}-\frac{\partial g_{00}}{\partial x_1}\right)+g^{2k}\left( \frac{\partial g_{02}}{\partial x_0} + \frac{\partial g_{20}}{\partial x_0}- \frac{\partial g_{00}}{\partial x_2} \right)+\dots \right]=\\
=&\frac{1}{2}\left[g^{0k}\left( \frac{\partial g_{00}}{\partial x_0}\right)+
g^{1k}\left( 2\frac{\partial g_{01}}{\partial x_0} -\frac{\partial g_{00}}{\partial x_1}\right)+g^{2k}\left( 2\frac{\partial g_{02}}{\partial x_0} - \frac{\partial g_{00}}{\partial x_2} \right)+\dots \right]=\\
=&\frac{1}{2}\left[g^{0k}\left( 2C+2AA'\right)+
g^{1k}\left( -2p'_1A-2p_1A'+2Ap_1'\right)+g^{2k}\left( -2p_2'A-2p_2A'+2Ap_2'\right)+\dots \right]=\\
=&\frac{1}{2}\left[g^{0k}\left( 2C+2AA'\right)+
g^{1k}\left(\cancel{ -2p'_1A}-2p_1A'+\cancel{2Ap_1'}\right)+g^{2k}\left( \cancel{-2p_2'A}-2p_2A'+\cancel{2Ap_2'}\right)+\dots \right]=\\
=&\frac{1}{\det g }\left[(p_{k}A\left( C+AA'\right)-p_1p_{k}B
\left( -p_1A' \right)-p_2p_{k}B\left( -p_2A'  \right)+\dots\right. \\
+&\left. \left[(1+p_1^2+\dots+p_{k-1}^2+p_{k+1}^2+\dots+p_{L-1}^2)B+A^2\right](-p_{k}A')+\dots-p_{k}p_{L-1}B\left( -p_{L-1}A' \right)\right]\\
=&\frac{p_{k}}{\det g }\left[AC-BA'\right]\\
\end{split}
\]	
}
\section{Coefficients of the curvature tensor $R$}\label{Rcoeff}
{\tiny
We have $$\langle R(X_0,X_i)X_0,X_i\rangle=\displaystyle \sum _s R^{s}_{0i0}g_{si}=R^0_{0i0}g_{0i}+R^1_{0i0}g_{1i}+\dots+R^{L-1}_{0i0}g_{(L-1)i},$$
where $$R_{0i0}^s=\displaystyle\sum_{m=0}^{L-1}\Gamma^m_{00}\Gamma_{im}^s-\displaystyle\sum_{m=0}^{L-1}\Gamma^m_{i0}\Gamma_{0m}^s+ \frac{\partial \Gamma_{00}^s}{\partial x_i} - \frac{\partial \Gamma^s_{i0}}{\partial x_0}.$$

Recall that  $\det g=(1+p_1^2+\dots+p_{L-1}^2)B+A^2= \Vert p\Vert^2B+A^2$ and
\[\begin{split}
\frac{\partial \det g}{\partial x_0}=&2\left[(p_1p_1'+p_2p_2'+\dots+p_{L-1}p_{L-1}')B+\Vert p\Vert ^2C+AA'\right]\\
\frac{\partial \det g}{\partial x_i}=&-2p_{i}'A\\
\end{split}
\]
The first addend is
$$R_{0i0}^0=\displaystyle\sum_{m=0}^{L-1}\Gamma^m_{00}\Gamma_{im}^0-\displaystyle\sum_{m=0}^{L-1}\Gamma^m_{i0}\Gamma_{0m}^0+ \frac{\partial \Gamma_{00}^0}{\partial x_i} - \frac{\partial \Gamma^0_{i0}}{\partial x_0},$$
that we can expand as
$$R_{0i0}^0=\cancel{\Gamma^0_{00}\Gamma_{i0}^0}+\cancel{\Gamma^1_{00}\Gamma_{i1}^0}+\cancel{\Gamma^2_{00}\Gamma_{i2}^0}+\dots-\cancel{\Gamma^0_{i0}\Gamma_{00}^0}-\Gamma^1_{i0}\Gamma_{01}^0-\Gamma^2_{i0}\Gamma_{02}^0+\dots+ \frac{\partial \Gamma_{00}^0}{\partial x_i} - \frac{\partial \Gamma^0_{i0}}{\partial x_0}.$$
We determine the derivative of Christoffel' symbols:
\[\begin{split}
\frac{\partial \Gamma^0_{00}}{\partial x_i}=&\frac{\partial }{\partial x_i}\frac{\left[(\Vert p\Vert^2)C+AA'\right]}{(\Vert p\Vert^2)B+A^2}=\\
=&\frac{(-p_{i}'A'-p_{i}''A)\left[(\Vert p\Vert^2)B+A^2\right]+2p_{i}'A\left[(\Vert p\Vert^2)C+AA'\right]}{(\det g)^2}\\
\end{split}
\] 
\[\begin{split}
\frac{\partial \Gamma^0_{i0}}{\partial x_0}=&\frac{\partial }{\partial x_0}\frac{-p_{i}'A}{(\Vert p\Vert^2)B+A^2}=\\
=&\frac{(-p_{i}''A-p_{i}'A')[(\Vert p\Vert^2)B+A^2]+2p_{i}'A\left[(p_1p_1'+p_2p_2'+\dots+p_{L-1}p'_{L-1})B+(\Vert p\Vert^2)C+AA'\right]}{(\det g)^2}\end{split}\]
By subtracting and simplifying,
$$\frac{\partial \Gamma_{00}^0}{\partial x_i} - \frac{\partial \Gamma^0_{i0}}{\partial x_0}=\frac{-2p_{i}'AB(p_1p_1'+\dots+p_{L-1}p'_{L-1})}{(\det g)^2}.$$
hence
\[
\begin{split}
R_{0i0}^0=&-\Gamma^1_{i0}\Gamma_{01}^0-\Gamma^2_{i0}\Gamma_{02}^0+\dots+ \frac{\partial \Gamma_{00}^0}{\partial x_i} - \frac{\partial \Gamma^0_{i0}}{\partial x_0}=\\
=&-\frac{p'_{i}p_1B}{\det g}\cdot\frac{-p_1'A}{\det g}-\frac{p'_{i}p_2B}{\det g}\cdot\frac{-p_2'A}{\det g}+\dots+\frac{-2p_{i}'AB(p_1p_1'+\dots+p_{L-1}p'_{L-1})}{(\det g)^2}=\\
=&\frac{p_{i}'AB(p_1p_1'+\dots+p_{L-1}p'_{L-1})}{(\det g)^2}-\frac{2p_{i}'AB(p_1p_1'+\dots+p_{L-1}p'_{L-1})}{(\det g)^2}=\\
=&-\frac{p_{i}'AB(p_1p_1'+\dots+p_{L-1}p'_{L-1})}{(\det g)^2}.\\
\end{split}\]
We calculate the coefficients for $k,i\neq 0$
$$R_{0i0}^k=\Gamma^0_{00}\Gamma_{i0}^k+\Gamma^1_{00}\cancel{\Gamma_{i1}^k}+\Gamma^2_{00}\cancel{\Gamma_{i2}^k}+\dots-\Gamma^0_{i0}\Gamma_{00}^k-\Gamma^1_{i0}\Gamma_{01}^k-\Gamma^2_{i0}\Gamma_{02}^k\dots+ \frac{\partial \Gamma_{00}^k}{\partial x_i} - \frac{\partial \Gamma^k_{i0}}{\partial x_0}.$$

\[\begin{split}
\frac{\partial \Gamma^k_{00}}{\partial x_i}=&\frac{\partial }{\partial x_i}\frac{p_{k}\left[AC-BA'\right]}{(\Vert p\Vert^2)B+A^2}=\\
=&\frac{p_{k}(-p_{i}'C+p''_{i}B)\left[(\Vert p\Vert^2)B+A^2\right]+2p_{k}p_{i}'\left[(AC-BA'\right]A}{(\det g)^2}=\\
=&\frac{-p_kp_{i}'C(\Vert p\Vert^2)B+p_kp''_{i}B(\Vert p\Vert^2)B-p_kp_{i}'CA^2+p_kp''_{i}BA^2+2p_{k}p_{i}'A^2C-2p_{k}p_{i}'BAA'}{(\det g)^2}\\
\end{split}
\]

\[\begin{split}
\frac{\partial \Gamma^k_{i0}}{\partial x_0}=&\frac{\partial }{\partial x_0}\frac{p_{i}'p_{k}B}{(\Vert p\Vert^2)B+A^2}=\\
=&\frac{(p_{k}p_{i}''B+p'_{i}p'_{k}B+2p_{i}'p_{k}C)[(\Vert p\Vert^2)B+A^2]-2p_{i}'p_{k}B\left[(p_1p_1'+p_2p_2'+\dots+p_{L-1}p'_{L-1})B+(\Vert p\Vert^2)C+AA'\right]}{(\det g)^2}\end{split}\]

\[\begin{split}
-\frac{\partial \Gamma^k_{i0}}{\partial x_0}=\\
=&\frac{-p_{k}p_{i}''B^2(\Vert p\Vert^2)-p'_{i}p'_{k}(\Vert p\Vert^2)B^2-2p_{i}'p_{k}C(\Vert p\Vert^2)B-p_{k}p_{i}''BA^2-p'_{i}p'_{k}BA^2-2p_{i}'p_{k}CA^2 }{(\det g)^2}\\
&+\frac{+2p_{i}'p_{k}B^2(p_1p_1'+p_2p_2'+\dots+p_{L-1}p'_{L-1})+2p_{i}'p_{k}B(\Vert p\Vert^2)C+2p_{i}'p_{k}BAA'}{(\det g)^2}.\end{split}\]

Hence the difference is
\[
\frac{\partial \Gamma_{00}^k}{\partial x_i} -\frac{\partial \Gamma^k_{i0}}{\partial x_0}=\frac{p_{i}'p_{k}\left[-BC\Vert p\Vert^2-CA^2+2B^2(p_1p_1'+p_2p_2'+\dots+p_{L-1}p'_{L-1})\right]-p'_{i}p'_{k}B[(\Vert p\Vert^2)B+A^2]}{(\det g)^2}
\]
and finally we obtain
\[\begin{split}
R_{0i0}^k=&\Gamma^0_{00}\Gamma_{i0}^k-\Gamma^0_{i0}\Gamma_{00}^k-\Gamma^1_{i0}\Gamma_{01}^k-\Gamma^2_{i0}\Gamma_{02}^k+\dots+ \frac{\partial \Gamma_{00}^k}{\partial x_i} - \frac{\partial \Gamma^k_{i0}}{\partial x_0}\\
=&\frac{[\Vert p\Vert^2C+AA']}{\det g}\frac{p'_{i}p_{k}B}{\det g}+\frac{p'_{i}A}{\det g}\frac{p_{k}\left[AC-A'B\right]}{\det g}-\frac{p'_{i}p_1B}{\det g}\frac{p'_{1}p_{k}B}{\det g}-\frac{p'_{i}p_{2}B}{\det g}\frac{p'_{2}p_{k}B}{\det g}+\dots+ \frac{\partial \Gamma_{00}^k}{\partial x_i} - \frac{\partial \Gamma^k_{i0}}{\partial x_0}=\\
=&\frac{p'_{i}p_{k}B[\Vert p\Vert^2C+AA']}{(\det g)^2}+\frac{p'_{i}p_{k}A\left[AC-A'B\right]}{(\det g)^2}-\frac{p'_{i}p_{k}B^2(p_1p_1'+\dots+p_{L-1}p'_{L-1})}{(\det g)^2}+ \frac{\partial \Gamma_{00}^k}{\partial x_i} - \frac{\partial \Gamma^k_{i0}}{\partial x_0}=\\
=&\frac{p_{i}'p_{k}B^2(p_1p_1'+p_2p_2'+\dots+p_{L-1}p'_{L-1})-p'_{i}p'_{k}B[(\Vert p\Vert^2)B+A^2]}{(\det g)^2}
\end{split}\]

Combining all the terms, we get
 \[\begin{split}
&\langle R(X_0,X_i)X_0,X_i\rangle=\displaystyle \sum _s R^{s}_{0i0}g_{si}=R^0_{0i0}g_{0i}+R^1_{0i0}g_{1i}+\dots+R^{l-1}_{0i0}g_{(L-1)i}=\\
=&\frac{p'_{i}AB(p_1p_1'+\dots+p'_{L-1}p_{L-1})}{(\det g)^2}\cdot p_{i}A+\frac{p_{i}'p_{1}B^2(p_1p_1'+\dots+p_{L-1}p'_{L-1})-p'_{i}p'_{1}B[(\Vert p\Vert^2)B+A^2]}{(\det g)^2}\cdot p_1p_{i}+\dots+\\
+&\frac{p_{i}'p_{i}B^2(p_1p_1'+\dots+p_{L-1}p'_{L-1})-p'_{i}p'_{i}B[(\Vert p\Vert^2)B+A^2]}{(\det g)^2}\cdot (1+p_{i}^2)+\dots + \\
+&\frac{p_{i}'p_{L-1}B^2(p_1p_1'+\dots+p_{L-1}p'_{L-1})-p'_{i}p'_{L-1}B[(\Vert p\Vert^2)B+A^2]}{(\det g)^2}\cdot p_{L-1}p_{i}=\\
=&\frac{p_{i}'p_{i}B(p_1p_1'+\dots+p_{L-1}p'_{L-1})\left[A^2+(\Vert p\Vert )^2B\right]-p'_{i}p'_{1}B[(\Vert p\Vert^2)B+A^2](p_1p_1'+\dots+p_{L-1}p'_{L-1})-(p_{i}')^2B[(\Vert p\Vert^2)B+A^2]}{(\det g)^2}=\\
=&\frac{-(p_{i}')^2B[(\Vert p\Vert^2)B+A^2]}{(\det g)^2}=\\
=&\frac{-(p_{i}')^2B}{(\det g)}.
\end{split} \]
}

\end{document}